\documentclass[letterpaper, 10pt, conference]{ieeeconf}
\IEEEoverridecommandlockouts                              

\usepackage{amssymb}
\usepackage{graphicx}
\usepackage{enumerate}
\usepackage{amsmath}
\usepackage{amsfonts}
\usepackage{algorithm}
\usepackage{algpseudocode}
\usepackage{url}
\usepackage{epstopdf}
\usepackage{multirow}
\usepackage{authblk}
\usepackage{xcolor}
\usepackage{epstopdf}
\usepackage{tabularx}
\usepackage{stmaryrd}
\usepackage{setspace}

\newtheorem{theorem}{Theorem}[section]
\newtheorem{example}[theorem]{Example}

\newtheorem{proposition}[theorem]{Proposition}
\newtheorem{definition}[theorem]{Definition}

\newcommand{\eps}{\varepsilon}

\graphicspath{{./fig/}}

\title{Supervisory Control with Event Forcing Under Partial Observation}

\author{Yin Tong$^{1}$ and Kai Cai$^{2}$
\thanks{This work was supported by the National Natural Science Foundation of China under Grant No. 62373313, by the Fundamental Research Funds for the Central Universities under Grant No. 2682025GH015, and by JST ASPIRE Grant no. JPMJAP2519.}
\thanks{$^1$Y. Tong (Corresponding Author) is with the School of Information Science and Technology, Southwest Jiaotong University, Chengdu 611756, China {\tt\small yintong@swjtu.edu.cn}.}%
\thanks{$^2$K. Cai is with the Department of Core Informatics, Osaka Metropolitan University, Osaka 558-8585, Japan {\tt\small cai@omu.ac.jp}.}
}

\begin{document}
\maketitle
\pagestyle{empty}

\begin{abstract}
In this paper, we address supervisory control in discrete-event systems with event forcing. In particular, partial observation of the supervisor is considered. Unlike traditional supervisory control, which relies on event enablement and disablement, the forcing mechanism allows a supervisor to preempt unwanted transitions by actively triggering forcible events. Under partial observation, we face a new challenge where the supervisor's forcing decisions must be consistent across all indistinguishable strings meanwhile guarantee all violating transitions are excluded through preemption. We introduce the concept of \emph{forcing consistency} as a necessary and sufficient condition for the existence of a supervisor that achieves a given specification. We prove that forcing consistency is strictly stronger than the notion of forcibility but, unlike the latter, is not closed under union. These results lay the theoretical foundation for supervisor synthesis in the presence of both forcing mechanisms and observation limitations.
\end{abstract}

\section{Introduction}\label{sec:intro}
Supervisory control theory (SCT) \cite{wonham2019supervisory}, introduced by Ramadge and Wonham in the 1980s, is a formal framework for the design of discrete-event systems (DES). In this paradigm, a supervisor controls a plant by enabling or disabling controllable events. The traditional SCT follows a simple rule: the supervisor decides which events are allowed, but the plant remains ``spontaneous'' in choosing which of those enabled events will actually occur. This approach has proven effective for a wide range of applications, including manufacturing systems, communication protocols, traffic control systems, and automated production lines.

Despite its theoretical elegance and practical success, traditional SCT still has some limitations in real-world control scenarios. The assumption that the plant makes the final choice does not always match how controllers work in practice. For instance, when a supervisor increases the power of an engine \cite{malikopoulos2014supervisory}, changes a traffic light from red to green \cite{du2021coupled}, or synthesize proteins from a gene \cite{baldissera2015supervisory}, these are mandatory commands that must be executed in a timely fashion, rather than optional permissions. Furthermore, Reijnen et al. \cite{reijnen2019finite,reijnen2022supervisory} showed that treating a controller as a passive filter can reduce system permissiveness. This mismatch between theory and practice motivates researchers to enrich the interaction mechanism.

The concept of event forcing has been introduced as an extension of the supervisory control paradigm. It allows the supervisor to mandate the execution of certain events, termed \emph{forcible events}. Unlike controllable events, which are merely enabled or disabled, forcible events are active commands and their execution preempts other transitions, including uncontrollable ones. This preemption mechanism is particularly valuable in safety-critical applications where timely intervention is essential. For example, in an emergency shutdown scenario, the supervisor must be able to force a halt operation that preempts normal operational events, ensuring that the system transitions immediately to a safe state regardless of other pending activities. Most research used event forcing mainly in timed DES to preempt the progress of time (e.g., \cite{brandin2002supervisory,zhang2013supervision,takai2006new,rashidinejad2020non,rashidinejad2023supervisory}), or in hybrid systems (e.g., \cite{ushio2005control}). 

In untimed DES, event forcing has received less attention. Initially, it was explored in \cite{golaszewski1987control}, which introduced a restricted controllability condition allowing for at most a single forced event. This was later extended in \cite{weidemann2012optimal} to accommodate multiple forced events through the use of control patterns. While some studies explored related concepts like driven events, enforceable events, or prioritized composition, they often lacked systematic synthesis methods \cite{golaszewski1987control,heymann2002concurrency} or suffered from computational complexity issues \cite{weidemann2012optimal,diekmann2013event}. Recently, Reniers and Cai \cite{reniers2024supervisory} formalized a comprehensive untimed SCT framework for event forcing. They introduced the notion of forcible-controllability, a property that is closed under union. Their work showed that this property provides a necessary and sufficient condition for a unique, maximally permissive supervisor that has the maximum flexibility in selecting  forcible events for preempting. They also provided a polynomial-time supervisor synthesis algorithm.

Current works on event forcing assume that a supervisor can observe the occurrence of all events. In practice, however, system dynamics may not be perfectly captured due to sensor limitations. In such scenarios, certain events become unobservable, leaving the supervisor with only partial observation. This setting is challenging because the supervisor must make identical forcing decisions for different strings that look the same.

In this paper, we extend the event-forcing framework of \cite{reniers2024supervisory} to the partial observation setting. In particular, we focus exclusively on event forcing to avoid the additional complication that arises when multiple control mechanisms overlap under partial observation. Our main contributions are:
\begin{itemize}
    \item A partial-observation supervisor and the corresponding closed-loop language are defined. Specifically, a control framework for event forcing is established where the supervisor's decisions are based solely on the observed event sequences.
    \item The notion of \emph{forcing consistency} is introduced, which serves as a counterpart of observability with event forcing. This property requires  uniform forcing decisions and preemptions across all indistinguishable strings. We prove that this property provides a necessary and sufficient condition for supervisor existence.
    \item The properties of forcing consistency are studied. We demonstrate that forcing consistency is not closed under union or under intersection, which implies that under partial observation, the unique maximally permissive supervisor generally does not exist.
\end{itemize}

The remainder of this paper is organized as follows. Section~\ref{sec:pre} recalls the basics of DES and traditional supervisory control. Section~\ref{sec:forcibility} introduces the event forcing mechanism, followed by the formulation of the partial observation framework and the corresponding definition of forcing consistency. The supervisor existence theorem and the properties of forcing consistency are presented in Section~\ref{sec:property}. Finally, Section~\ref{sec:conclusion} concludes the paper.
\section{Preliminaries}\label{sec:pre}

In this section, we recall some basics on automata and supervisory control theory. For further details, the reader is referred to \cite{cassandras2021introduction}.

\subsection{Automaton}
The plant is modeled by a \emph{deterministic finite automaton} (DFA) $$G=(X, E, f, x_0,X_m),$$ where $X$ is the finite set of \emph{states}, $E$ is the finite set of \emph{events}, $f: X \times E \rightarrow X$ is the \emph{transition function}, $x_0 \in X$ is the \emph{initial state}, and $X_m\subseteq X$ is the set of \emph{marked states}. The transition function can be extended to $f: X\times E^*\rightarrow X$ in the usual way. Given $x\in X$ and $s\in E^*$, we denote $f(x,s)!$ that transition $f(x,s)$ is defined. The \emph{language generated} by $G$ is defined as $$L(G):=\{s\in E^*|f(x_0,s)!\}.$$ 
The prefix closure of language $L\subseteq E^*$ is $\overline{L}:=\{s\in E^*|(\exists s'\in E^*) ss'\in L\}$.

The set of events is partitioned into $E_o \subseteq E$ the set of events \emph{observable} by the supervisor and $E_{uo}=E\setminus E$ the set of \emph{unobservable} events. The \emph{natural projection} $P:E^* \rightarrow E^*_o$ is defined as i) $P(\eps):=\eps$; ii) for all $s\in E^*$ and $e\in E$, $P(se)=P(s)e$ if $e\in E_o$, and $P(se)=P(s)$, otherwise. Given a language $L\subseteq E^*$, its projection on $E_o$ is defined as $P(L):=\bigcup_{s\in L}\{P(s)\}$. Given an observation sequence $w\in E^*_o$, the set of strings in $L(G)$ whose projection on $E^*_o$ equal to $w$ is defined as $$P^{-1}(w):=\{s\in L(G)|P(s)=w\}.$$

Given a string $s\in L(G)$ and a sublanguage $K\subseteq L(G)$, we denote $$\Gamma(s):=\{e\in E|f(x_0,se)!\}$$ 
the set of events \emph{active} after $s$ in $G$; 
$$\Gamma_K(s):=\{e\in E|se\in \overline{K}\}$$ the set of events whose occurrences after $s$ is \emph{within} $\overline{K}$. Clearly, $\Gamma_K(s)\subseteq \Gamma(s)$, and $\Gamma(s)\setminus \Gamma_K(s)\neq \emptyset$ if and only if $\Gamma(s)\neq \Gamma_K(s)$.

\subsection{Supervisory control}
In traditional supervisory control (Ramadge-Wonham framework \cite{ramadge1989control}), the set $E$ of events is partitioned into the set $E_c$ of \emph{controllable} events whose occurrences can be disabled by the supervisor and the set $E_{uc}=E\setminus E_c$ of \emph{uncontrollable} events whose occurrences cannot be disabled. Given a specification $K\subseteq L(G)$, a supervisor $S:E_o^*\rightarrow 2^E$ under partial observation is designed to restrict the behavior of the plant: after observing $w\in P(L(G))$, the supervisor sends to the plant its \emph{control decision}, a subset of $E$ containing all uncontrollable events that can occur next. The closed-loop system of $G$ supervised by $S$ is denoted by $S/G$. The behsvior of $S/G$, denoted by $L(S/G)$, is formally defined as follows:
\begin{itemize}
    \item $\eps \in L(S/G)$; and 
    \item $se\in L(S/G) \Leftrightarrow s\in L(S/G),e\in S(P(s)), \text{ and } se\in L(G)$.
\end{itemize}

According to supervisory control theory, there exists a supervisor $S$ for $G$ such that the closed-loop behavior $L(S/G)=\overline{K}$ if and only if the specification $K$ is \emph{observable} and \emph{controllable} w.r.t $G$, $E_c$, and $P$.  

\section{Supervisory Control with Event Forcing}\label{sec:forcibility}
In the traditional supervisory control, the interaction between a supervisor and a plant is limited to the mechanism of enablement and disablement. Recently, work \cite{reniers2024supervisory} enriched this paradigm by introducing the concept of \emph{event forcing}. In this section, we first recall the definitions and results of event forcing and then extend the framework to supervisor under partial observation. Based on the extended framework, we define the notion of forcing consistency and discuss its relation with forcibility.

\subsection{Event forcing}
In forcing framework, the set $E$ of events is partitioned into $E_f\subseteq E$ the set of \emph{forcible} events and $E_{uf}=E\setminus E_f$ the set of \emph{unforcible} events. A forcible event ($f \in E_f$) is an event that can be forced by an external agent (e.g., supervisor) to occur, thereby preempting all other active events (including uncontrollable ones). The supervisor's control decision may include more than one forcible event for the sake of maximal freedom of choice; however, in implementation, the forcing agent nondeterministically selects exactly one forcible event from this set to execute and preempt all other events.

In this paper, we focus exclusively on event forcing. Specifically, the supervisor's only means of intervention is through the forcing of events, while the classical mechanisms of enablement and disablement are not considered. This allows us to isolate and study the unique challenges that partial observation imposes on the forcing mechanism. An extension of this work to include both traditional controllability and event forcibility will be further studied in future work. 


\begin{definition}\cite{reniers2024supervisory}\label{def:forciblity}
Given a plant $G$ and a set of forcible events $E_f$, a language $K\subseteq L(G)$ is \emph{forcible} (w.r.t $G$ and $E_f$) if $\forall s\in \overline K$,
\begin{align}
    & \Gamma(s)=\Gamma_K(s) \label{eq:F-1}\\ 
    & \vee \nonumber\\ 
    & (\Gamma_K(s)\cap E_f\neq \emptyset \wedge \Gamma_K(s)\setminus E_f=\emptyset). \label{eq:F-2}
\end{align}
\hfill $\diamond$
\end{definition}
In simple words, forcibility means that for any string $s$, its continuations must either already satisfy the specification (condition~\eqref{eq:F-1}), or there is at least one forcible event that retains the evolution within $\overline{K}$ while no unforcible events stays in $\overline{K}$ (condition~\eqref{eq:F-2}).

As proven in \cite{reniers2024supervisory}, forcibility is necessary and sufficient for the existence of a supervisor $S$ that guarantees $L(S/G) = \overline{K}$. Furthermore, the property of forcibility is \emph{closed under union}, ensuring the existence of a unique supremal forcible sublanguage for any given specification.

\subsection{Supervisor under partial observation}\label{sec:sup}
While the foundational work in \cite{reniers2024supervisory} assumes that the plant $G$ is a DFA and the supervisor possesses full observation (i.e., all events in $E$ are observable), real-world applications often involve sensors that can only detect a subset of events. In the following, we extend the theory of supervisory control with event forcing to the partial observation setting, where the supervisor must make forcing decisions based on the projection of system's evolutions.

The supervisor with events forcing under partial observation is defined as a function $S: P(L(G))\rightarrow 2^{E_f}$. It means that given an observation $w\in P(L(G))$ of the supervisor, the set of events forced to occur is $S(w)\subseteq E_f$. In this paper, no assumption is made on the relation between $E_o$ and $E_f$. Thus, there may exist events that are forcible but not observable, i.e., $E_f\cap E_{uo}\neq\emptyset$.

\begin{definition}\label{def:S/G}
The closed-loop system of $G$ under the control of partial-observation supervisor $S$ is denoted as $S/G$. The language generated by $S/G$ is recursively defined as follows:
\begin{itemize}
    \item $\eps\in L(S/G)$;
    \item $\forall s\in L(S/G)$ and $e\in \Gamma(s)$, ($S(P(s))\cap \Gamma(s)=\emptyset$ $\vee$ $e\in S(P(s))$) $\Leftrightarrow$ $se\in L(S/G)$. \hfill $\diamond$
\end{itemize}
\end{definition}

The condition $S(P(s))\cap \Gamma(s)=\emptyset$ implies that the supervisor does not force any events, and thus all active events $e$ can occur after $s$. In contrast, the condition $e\in S(P(s))$ means that event $e$ must occur after $s$. Consequently, any other active event not in $S(P(s))$ is preempted.

It is important to notify that the language $L(S/G)$ defined in Definition \ref{def:S/G} represents the collection of all possible event sequences that \emph{could potentially be generated under the control} of $S$. However, $L(S/G)$ may not be physically realized as a whole. This is because unlike traditional supervisory control with events disabling, in the forcing framework, the occurrence of any event $e \in S(P(s))$ is exclusive: if multiple events are candidates for forcing, only one of them is selected to execute, thereby preventing all other active events from occurring at that instant. Consequently, while $L(S/G)$ captures all the possible behaviors of the closed-loop system for the purpose of formal analysis, any actual physical run of the system will only result in a specific path within this language.

\subsection{Forcing consistency}\label{sec:FC}
In traditional supervisory control, the occurrence of an event is solely determined by the supervisor's enablement or disablement. Specifically, an event $e$ is prevented from occurring if and only if it is disabled (i.e., not included in the control decision). As originally introduced by Lin and Wonham \cite{lin1988observability}, the concept of \emph{observability} requires that for any two strings $s, s' \in \overline{K}$ sharing the same observation $P(s) = P(s')$, the supervisor's decision to enable or disable a controllable event must be consistent. Formally, for any controllable event $e$, it holds that $e \in \Gamma_K(s) \cap \Gamma(s')\Rightarrow e\in \Gamma_K(s')$. Under this framework, the control decision for an observation $P(s)$ is straightforward: it simply excludes all active events that would lead any $s\in P^{-1}(w)$ outside $K$. If $K$ is both controllable and observable, the resulting supervisor ensures that the closed-loop behavior is identical to $\overline{K}$.

However, when the control mechanism shifts from passive disablement to active event forcing, 
the non-occurrence of an event $e$ is no longer a simple consequence of its exclusion from the control decision. Instead, event $e$ may be preempted because another event was forced. This shift makes it difficult to formulate a suitable counterpart to observability. Consequently, establishing a necessary and sufficient condition for the existence of a supervisor such that $L(S/G) = \overline{K}$ becomes a non-trivial extension of classical theory. Examples~\ref{eg:motivation} and \ref{eg:SComputation} illustrate these challenges.

\begin{figure}
    \centering
    \includegraphics[width=0.35\textwidth]{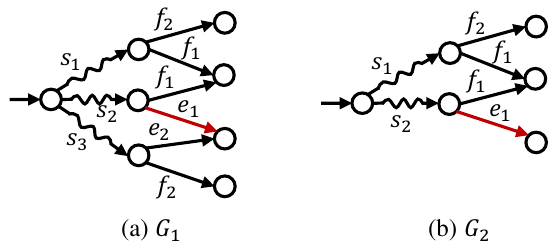}
    \caption{Plants $G_1$ and $G_2$, where $E=E_o=\{e_1,e_2,f_1,f_2\}$, $E_f=\{f_1,f_2\}$, and $s_1,s_2,s_3\in E^*$. Let $K_1=L(G_1)\setminus\{s_2e_1\}$ and $K_2=L(G_2)\setminus \{s_2e_1\}$.}\label{fig:G1G2}
\end{figure}

\begin{example}\label{eg:motivation}
 In the two plants shown in Fig.~\ref{fig:G1G2}, assume all events $E=\{e_1,e_2,f_1,f_2\}$ are observable, and $s_i\in E^*$ satisfies $P(s_i)=w$ for all $i\in \{1,2,3\}$. Let $f_1,f_2$ be forcible events and $e_1,e_2$ be unforcible. Consider plant $G_1$ in Fig.~\ref{fig:G1G2}(a) and the specification $K_1=L(G_1)\setminus\{s_2e_1\}$. According to Definition~\ref{def:forciblity}, $K_1$ is forcible w.r.t $G_1$ and $E_f=\{f_1,f_2\}$. Forcing $f_1$ is necessary to prevent $e_1$ from occurring after $s_2$. Moreover, for any $i\in \{1,2,3\}$, $s_if_1\in \overline{K}$. Namely, there is no conflict on forcing or not forcing $f_1$ after observing $w$. Thus, one possible control decision is $S(w)=\{f_1\}$.  However, this decision also preempts event $f_2$ after $s_1$, leading to $L(S/G)=L(G)\setminus\{s_2e_1,s_1f_2\}\neq \overline{K_1}$. Alternatively, the supervisor could force both $f_1$ and $f_2$ after observing $w$, i.e., $S(w)=\{f_1,f_2\}$. In this case, event $e_2$ after $s_3$ is preempted. As a result, $L(S/G_1)=L(G_1)\setminus\{s_2e_1,s_3e_2\}\neq\overline{K_1}$. In fact, a control decision $S(w)\subseteq E_f$ that guarantees $L(S/G_1)=\overline{K_1}$ does not exist. \hfill $\diamond$
\end{example}

Example~\ref{eg:motivation} demonstrates that even if (i) $K$ is forcible and (ii) the supervisor can make consistent forcing decisions for each forcible event (i.e., for all $s_i,s_j\in P^{-1}(w)$ and $f\in E_f$, $f\in \Gamma_K(s_i) \Rightarrow f\in \Gamma_K(s_j)$), a supervisor $S$ satisfying $L(S/G) = \overline{K}$ may still not exist.

In traditional supervisory control, when observing $w$, the supervisor's decision is a union of safe actions across $P^{-1}(w)$ excluding those that are unsafe for some $s\in P^{-1}(w)$. To precisely characterize these interactions in event forcing, we define the following sets of forcible events for any $s\in L(G)$ over a given $K\subseteq L(G)$: 
\begin{itemize}
    \item $F^+(s)=\left\{\begin{array}{ll}
        E_f \cap \Gamma_K(s), & \text{ if } \Gamma(s) \setminus \Gamma_K(s)\neq\emptyset;\\
        \emptyset, & \text{ otherwise};
    \end{array}
        \right.$

        the set of forcible events that can be forced to preempt transitions leaving $\overline{K}$; 
    \item $F^-(s)=E_f\cap(\Gamma(s)\setminus \Gamma_K(s))$: the set of forcible events that must not be forced because their occurrence violates $\overline{K}$.
\end{itemize}
 The following example shows that the control decision for $P(s)=w$ cannot be obtained through simple set operations on $F^+(s)$, $F^-(s)$ across all $s\in P^{-1}(w)$, as is typically done in classical supervisory control.

\begin{example}\label{eg:SComputation}[Continuing Example~\ref{eg:motivation}]
Consider plant $G_2$ in Fig.~\ref{fig:G1G2}(b) and let $K_2=L(G_2)\setminus\{s_2e_1\}$, which is forcible w.r.t $G_2$ and $E_f$. Following the classical logic of ``enable what is necessary and disable what is unsafe,'' an intuitive control decision for $w$ would be $S(w) = (\bigcup_{s \in P^{-1}(w)} F^+(s)) \setminus (\bigcup_{s \in P^{-1}(w)} F^-(s)) = \{f_1\}$. 
However, this decision results in $L(S/G_2)=L(G_2)\setminus\{s_2e_1,s_1f_2\}\neq \overline{K_2}$. To exactly enforce $\overline{K_2}$, the supervisor must actually include $f_2$ in its decision ($S(w)=\{f_1,f_2\}$), even though $f_2$ is not required for preemption after $s_2$. \hfill $\diamond$
\end{example}

These examples demonstrate that achieving $L(S/G) = \overline{K}$ requires more than just preventing transitions exiting $K$. The interdependence between forcing decisions and event preemptions significantly complicates supervisor synthesis. This suggests that a proper notion of ``observability'' for event forcing must account for these preemptive effects, leading to a definition that is less intuitive than the classical case. Based on these observations, we formalize the requirements for $K$ through the following notion of \emph{forcing consistency}, which provides a necessary and sufficient condition for the existence of such a supervisor.

\begin{definition}[Forcing Consistency]\label{def:FC}
    Given a plant $G$, a set of forcible events $E_f$, and the projection $P$, a sublanguage $K\subseteq L(G)$ is \emph{forcing consistent} (w.r.t. $G$, $E_f$, and $P$) if $\forall w\in P(\overline K)$, there exists a subset $E^\text{cd}_f(w)\subseteq E_f$ of forcible events such that $\forall s\in P^{-1}(w)\cap \overline K$,
    \begin{align}
        & (\Gamma(s)\cap E^\text{cd}_f(w)\neq \emptyset  \Rightarrow \Gamma_K(s)= E^\text{cd}_f(w)\cap \Gamma(s)) \label{eq:FC-1}\\
        & \vee \nonumber  \\
        & (\Gamma(s)\cap E^\text{cd}_f(w) = \emptyset  \Rightarrow \Gamma(s)=\Gamma_K(s)). \label{eq:FC-2}
    \end{align}
        \hfill $\diamond$ 
\end{definition}

Forcing consistency requires that for any observation $w$, there must exist a corresponding common set $E^\text{cd}_f(w)$ of forcible events for all $s\in P^{-1}(w)$. This set ensures two conditions: first, for $s$ satisfying $\Gamma(s)\cap E^\text{cd}_f(w)\neq \emptyset$, the decision to force or not force an event in $E_f$ is consistent --- forcing any event in $E^\text{cd}_f(w)$ would keep the continuation of $s$ within $\overline{K}$ while forcing any event in $E_f\setminus E^\text{cd}_f(w)$ would lead a string $s$ violating $\overline{K}$; second, any event preempted by this forcing action must be one that would have otherwise led the system outside $\overline{K}$. The condition \eqref{eq:FC-2} is needed to ensure that when nothing is forced after some $s$, no event can lead $s$ outside $\overline{K}$. 
\begin{figure}
    \centering
    \includegraphics[width=0.48\textwidth]{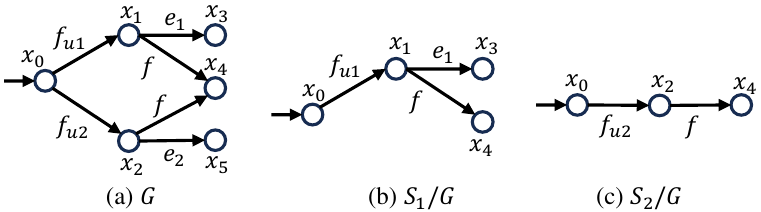}
    \caption{Plant $G$, where $E_f=\{f_{u1},f_{u2},f\}$ and $E_o=\{f,e_1,e_2\}$. Two forcing consistent sublanguages $K_1=L(S_1/G)$ and $K_2=L(S_2/G)$.}\label{fig:NotClose}
\end{figure}

\begin{example}\label{eg:FC}
    Let us revisit Example~\ref{eg:motivation}. For observation $w$ in $G_1$, neither conditions \eqref{eq:FC-1} nor \eqref{eq:FC-2} holds for any possible value of $E^\text{cd}_f(w)\subseteq E_f=\{f_1,f_2\}$. According to Definition~\ref{def:FC}, $K_1$ is not forcing consistency w.r.t $G_1$, $E_f$, and $P$.

    Consider plant $G$ in Fig.~\ref{fig:NotClose}(a) and two sublanguages $K_1=\overline{\{f_{u1}f,f_{u1}e_1\}},K_2=\overline{\{f_{u2}f\}}$ of $L(G)$. The sets of forcible and observable events are $E_f=\{f,f_{u1},f_{u2}\}$ and $E_o=\{f,e_1,e_2\}$, respectively. Both $K_1$ and $K_2$ are forcing consistent w.r.t $G$, $E_f$, and $P$. Indeed,
    \begin{itemize}
    \item For $K_1$, $\forall w\in P(\overline{K_1})$, sets $E_f^{\text{cd}}(w)\subseteq E_f$ satisfying conditions~\eqref{eq:FC-1} or \eqref{eq:FC-2} exist: 
    $$E_f^{\text{cd}}(w)=\left\{\begin{array}{ll}
    \{f_{u1}\}, & \text{ if } w=\eps;\\
    \emptyset, &\text{ otherwise}.
    \end{array}\right.$$
    \item For $K_2$, $\forall w\in P(\overline{K_2})$, sets $E_f^{\text{cd}}(w)\subseteq E_f$ satisfying conditions~\eqref{eq:FC-1} or \eqref{eq:FC-2} exist:
    $$E_f^{\text{cd}}(w)=\left\{\begin{array}{ll}
    \{f,f_{u2}\}, & \text{ if } w=\eps;\\
    \emptyset, &\text{ otherwise}.
    \end{array}\right.$$
    \end{itemize}
\hfill $\diamond$
\end{example}

In traditional supervisory control, controllability and observability are independent properties: one does not imply the other. Therefore, $K$ must satisfy both properties to guarantee the existence of a supervisor. Forcing consistency, however, is a \emph{strictly stronger} property than forcibility. Specifically, forcing consistency extends the requirement of forcibility by imposing an observational constraint: it requires not only the existence of a valid forcing action to keep the system within $\overline{K}$, but also that this choice remains uniform and valid across all indistinguishable strings.
Next we formally prove that if a language $K$ is forcing consistent, it is inherently forcible. The converse is not necessarily true.

\begin{proposition}\label{prop:FC-F}
     Given a plant $G$, a set of forcible events $E_f$, and the projection $P$, if a language $K\subseteq L(G)$ is forcing consistent, it is also forcible.
\end{proposition}
\begin{proof}
    Suppose, for the sake of contradiction, that $K$ is forcing consistent but not forcible. Then there exists a string $s\in \overline{K}$ that fails to satisfy both conditions~\eqref{eq:F-1} and \eqref{eq:F-2}. Let $w=P(s)$, and let $E^{cd}_f(w)\subseteq E_f$ be the set of forcible events that ensure any $s\in P^{-1}(w)$ satisfies condition~\eqref{eq:FC-1} or \eqref{eq:FC-2}. We consider the following two cases:
    
    Case 1: $\Gamma(s)\cap E_f^{cd}(w)\neq \emptyset$. Since $K$ is forcing consistent, we have $\Gamma_K(s)=\Gamma(s)\cap E_f^{cd}(w)$. It follows that $\Gamma_K(s)\subseteq E^{cd}_f(w)$. Given that $E_f^{cd}(w)\subseteq E_f$,  this implies $\Gamma_K(s)\cap E_f\neq \emptyset$ and $\Gamma_K(s) \setminus E_f=\emptyset$. Thus, string $s$ satisfies condition~\eqref{eq:F-2}, which is a contradiction.

    Case 2: $\Gamma(s)\cap E_f^{cd}(w)= \emptyset$. Since $K$ is forcing consistent, $\Gamma_K(s)=\Gamma(s)$. In this case, $s$ satisfies condition~\eqref{eq:F-1}, which again leads to a contradiction.

    The converse of the proposition does not hold. For instance, the language $K_1$ in Example~\ref{eg:motivation} is forcible but not forcing consistent.
\end{proof}

To verify if a regular language $K$ is forcing consistent, or to determine the sets $E^{cd}_f(w)$ for a given forcing consistent $K$, a straightforward method is to use an observer-based approach (similar to the method in Sec.~3.7.1 of \cite{cassandras2021introduction}). By transforming the language $K$ into a state specification and constructing the observer of the automaton generating $\overline{K}$, the required conditions can be verified over a finite state set.

Each state in the observer represents a set of plant states $X'\subseteq X$ consistent with the observation $w \in P(\overline{K})$. For each $x\in X'$, one can test all $2^{|E_f|}$ subsets of $E_f$. A subset is valid if it satisfies condition~\eqref{eq:FC-1} or \eqref{eq:FC-2} for the active event set at $x$ and the safe active event set at $x$. While this exhaustive search is not optimized, it provides a direct way to verify the property and determine $E^{cd}_f(w)$.  Examples are provided below to illustrate the idea.
\begin{example}\label{eg:test}
    Consider again plant $G$ in Fig.~\ref{fig:NotClose}(a) and let $K=L(G)\setminus\{f_{u2}e_2\}$. The corresponding state specification of $K$ is defined by the set of safe states $X_K=X\setminus \{x_5\}$. The automaton $G_K$ generating $\overline{K}$ and its observer are depicted in Fig.~\ref{fig:GK}.     
    For $w=\eps$, the set of states in $G_K$ consistent with the observation is $y_0=\{x_0,x_1,x_2\}$. The sets of active events at each consistent state in $G$ are $\Gamma(x_0)=\{f_{u1},f_{u2}\}$, $\Gamma(x_1)=\{f,e_1\}$, and $\Gamma(x_2)=\{f,e_2\}$. The sets of safe active events (those not reaching state $x_5$) at each consistent state in $G$ are $\Gamma_K(x_0)=\Gamma(x_0)$, $\Gamma_K(x_1)=\Gamma(x_1)$, and $\Gamma_K(x_2)=\Gamma(x_2)\setminus\{e_2\}$. 
    By testing every subset of $E_f$, we find that none of them satisfy condition~\eqref{eq:FC-1} or \eqref{eq:FC-2} (with $\Gamma(s)$ and $\Gamma_K(s)$  replaced by $\Gamma(x_i)$ and $\Gamma_K(x_i)$) for all $x_i\in y_0$. Thus, no such set $E_f^{cd}(\eps)\subseteq E_f$ exists, implying that $K$ is not forcing consistent.

    Now consider $K_1$ in Fig.~\ref{fig:NotClose}(b). The automaton $G_{K_1}$ generating $\overline{K_1}$ is identical to the one in Fig.~\ref{fig:NotClose}(b). For $w=\eps$, the set of consistent states in $G_{K_1}$ is $\{x_0,x_1\}$. For each state here, $\Gamma(x_0)=\{f_{u1},f_{u2}\}$ and $\Gamma(x_1)=\{f,e_1\}$, while $\Gamma_K(x_0)=\Gamma(x_0)\setminus\{f_{u2}\}$ and $\Gamma_K(x_1)=\Gamma(x_1)$. The only subset of $E_f$ that ensures $\Gamma(x_i)$ and $\Gamma_K(x_i)$ satisfy condition~\eqref{eq:FC-1} or \eqref{eq:FC-2} for all $x_i\in \{x_0,x_1\}$ is $\{f_{u1}\}$. Therefore, $E^{cd}_f(\eps)=\{f_{u1}\}$. For $w=e_1$ (resp., $f$) in $P(\overline{K})$,  any subset of $E_f$ satisfies the conditions for the consistent state $x_3$ (resp., $x_4$); thus, $E^{cd}_f(e_1)$ (resp., $E^{cd}_f(f)$) can be any subset of $E_f$. It follows that $K_1$ is forcing consistent, and the sets $E^{cd}_f(w)$ are determined for all $w\in P(\overline{K_1})$. \hfill $\diamond$
\end{example}

\begin{figure}
    \centering
    \includegraphics[width=0.32\textwidth]{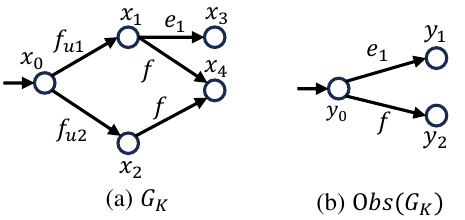}
    \caption{(a) Automaton $G_K$ with $L(G_K)=\overline{K}$ and (b) its observer, where $y_0=\{x_0,x_1,x_2\}$, $y_1=\{x_3\}$, and $y_2=\{x_4\}$.}\label{fig:GK}
\end{figure}

\section{Supervisor Existence and Property of Forcing Consistency}\label{sec:property}

\subsection{Supervisory existence}\label{sec:NecAndSuf}
Based on the definition of forcing consistency and its relation with forcibility, we now establish the necessary and sufficient condition for the existence of a partial-observation supervisor that achieves the specification $K$ via event forcing.

\begin{theorem}\label{thm:S}
Given a plant $G$, a set of forcible events $E_f$, the projection $P$, and a language $\emptyset\neq K\subseteq L(G)$, there exists a partial-observation supervisor $S:P(L(G))\rightarrow 2^{E_f}$ such that $L(S/G)=\overline K$ if and only if $K$ is forcing consistent w.r.t. $G$, $E_f$, and $P$.
\end{theorem}
\begin{proof}
(If) Suppose that $K$ is forcing consistent w.r.t $G$, $E_f$, and $P$. Then, By Definition~\ref{def:FC}, the forcing consistency of $K$ implies that for each observation $w \in P(\overline{K})$, there exists at least one subset of forcible events $E^\text{cd}_f(w) \subseteq E_f$ that satisfies either condition~\eqref{eq:FC-1} or \eqref{eq:FC-2}. Furthermore, $K$ is forcible w.r.t $G$ and $E_f$ according to Proposition~\ref{prop:FC-F}.
We now construct a supervisor $S:P(L(G))\rightarrow 2^{E_f}$ by using these predefined sets $E^{cd}_f(w)$ as follows: $\forall w\in P(L(G))$,
    \begin{itemize}
        \item If ($\forall s\in P^{-1}(w)\cap \overline K$) $\Gamma_K(s)\cap E_f=\emptyset$ $\vee$ $\Gamma(s)=\Gamma_K(s)$, 
            \begin{equation}\label{eq:S(w)-1}
                S(w):=\emptyset.
            \end{equation}
        \item Otherwise,
            \begin{equation}\label{eq:S(w)-2}
                S(w):=E^\text{cd}_f(w).
            \end{equation}
    \end{itemize}
In other words, if no forcible event occurs after any consistent string $s$ or no event can lead $s$ outside $\overline K$, then no forcible event is enforced by the supervisor. Otherwise (i.e., ($\exists s\in P^{-1}(w)\cap \overline K$) $\Gamma_K(s)\cap E_f\neq \emptyset$ $\wedge$ $\Gamma(s)\setminus \Gamma_K(s)\neq \emptyset$), forcible events in $E^\text{cd}_f(w)$, if the set is not empty, are forced.  

We claim that under the control of the supervisor defined above $$L(S/G)=\overline K.$$ We prove this claim by induction on the length of strings $s\in L(G)$. For the base case, as $K\neq \emptyset$, $\eps\in \overline K$. By Definition~\ref{def:S/G}, $\eps\in L(S/G)$. Assume the induction hypothesis that for a string $s\in E^*$, $$s\in L(S/G)\Leftrightarrow s\in \overline{K}.$$ Let $P(s)=w$ and $e\in \Gamma(s)$. We first prove that $se\in L(S/G)\Rightarrow se\in \overline K$. By the hypothesis, $s\in \overline K$. We consider two cases.

Case 1: ($\forall \hat s\in P^{-1}(w)\cap \overline K$) $\Gamma_K(\hat s)\cap E_f=\emptyset$ $\vee$ $\Gamma(\hat s)=\Gamma_K(\hat s)$. In this case, if $\Gamma(s)=\Gamma_K(s)$, $se\in \overline{K}$. If $\Gamma_K(s)\cap E_f=\emptyset$, as $K$ is forcible, Eq.~\eqref{eq:F-1} must hold. Thus, $\Gamma_K(s)=\Gamma(s)$ and $se\in \overline{K}$.

Case 2: ($\exists \hat s\in P^{-1}(w)\cap \overline K$) $\Gamma_K(\hat s)\cap E_f\neq \emptyset$ $\wedge$ $\Gamma(\hat s)\setminus \Gamma_K(\hat s)\neq\emptyset$. Since $se\in L(S/G)$, by Definition~\ref{def:S/G}, two possibilities: (i) $S(w)\cap \Gamma(s)=\emptyset$ and (ii) $e\in S(w)$, are considered. For (i), $S(w)\cap \Gamma(s)=E^\text{cd}_f(w)\cap \Gamma(s)=\emptyset$. Since $K$ is forcible, Eq.~\eqref{eq:F-1} holds and $\Gamma(s)=\Gamma_K(s)$. Thus, $se\in \overline K$. For (ii), it implies that $S(w)\cap \Gamma(s)=E^\text{cd}_f(w)\cap \Gamma(s)\neq \emptyset$. As $K$ is forcing consistent, condition~\eqref{eq:FC-1} holds and $\Gamma_K(s)=E^\text{cd}_f(w)\cap \Gamma(s)=S(w)\cap \Gamma(s)$. Thus, $e\in \Gamma_K(s)$ and $se\in \overline K$.

In both cases above, we established $se\in L(S/G)\Rightarrow se\in \overline K$. Now we prove $se\in \overline K\Rightarrow se\in L(S/G)$. Since $\overline K$ is prefix-closed, $s\in \overline K$. By the hypothesis, $s\in L(S/G)$. To show that $se\in L(S/G)$, by Definition~\ref{def:S/G}, we need to prove either $e\in S(w)$ or $S(w)\cap \Gamma(s)=\emptyset$. We consider again the two cases above.

Case 1: ($\forall \hat s\in P^{-1}(w)$) $\Gamma_K(\hat s)\cap E_f=\emptyset$ $\vee$ $\Gamma(\hat s)=\Gamma_K(\hat s)$. In this case, by Eq.~\eqref{eq:S(w)-1}, $S(w)=\emptyset$. As a result, $S(w)\cap \Gamma(s)=\emptyset$.

Case 2: ($\exists \hat s\in P^{-1}(w)$) $\Gamma_K(\hat s)\cap E_f\neq \emptyset$ $\wedge$ $\Gamma(\hat s)\setminus \Gamma_K(\hat s)\neq\emptyset$. In this case, by Eq.~\eqref{eq:S(w)-2}, $S(w)=E^\text{cd}_f(w)$. Since $se\in \overline K$, $e\in \Gamma_K(s)$. Suppose that $e\notin S(w)$ and $S(w)\cap \Gamma(s)\neq \emptyset$. Since $K$ is forcing consistent and $S(w)\cap \Gamma(s)\neq \emptyset$, by condition~\eqref{eq:FC-1}, $\Gamma_K(s)=E^\text{cd}_f(w)\cap \Gamma(s)=S(w)\cap \Gamma(s)$. Thus, $e\in S(w)$, which is a contradiction.

In both cases above, we prove $se\in \overline{K}\Rightarrow se\in L(S/G)$. Having shown both directions, we conclude that $L(S/G)=\overline K$.

(Only if) Suppose that there exists a supervisor $S:P(L(G))\rightarrow 2^{E_f}$ such that $L(S/G)=\overline K$. We show that $K$ is forcing consistent w.r.t $G$, $E_f$, and $P$. Let $w\in P(\overline K)$ and $s\in P^{-1}(w)\cap \overline K$.
Consider the two cases again.

Case 1: ($\forall \hat s\in P^{-1}(w)\cap \overline K$) $\Gamma_K(\hat s)\cap E_f=\emptyset$ $\vee$ $\Gamma(\hat s)=\Gamma_K(\hat s)$. By Eq.~\eqref{eq:S(w)-1}, $S(w)=\emptyset$. Let $E_f^{cd}(w)=S(w)$. By Definition~\ref{def:S/G}, $\forall e\in \Gamma(s)$, $se\in L(S/G)=\overline K$. This implies that $\Gamma(s)=\Gamma_K(s)$. Therefore, condition~\eqref{eq:FC-2} is satisfied for all $s\in P^{-1}(w)$. 

Case 2: ($\exists \hat s\in P^{-1}(w)\cap \overline K$) $\Gamma_K(\hat s)\cap E_f\neq \emptyset$ $\wedge$ $\Gamma(\hat s)\setminus \Gamma_K(\hat s)\neq\emptyset$. In this case, let $E^\text{cd}_f(w)=S(w)$. We consider the two possibilities again: (i) $S(w)\cap \Gamma(s)=\emptyset$ and (ii) $S(w)\cap \Gamma(s)\neq \emptyset$. For (i), $E_f^\text{cd}(w)\cap \Gamma(s)=\emptyset$. By Definition~\ref{def:S/G}, $\forall e\in \Gamma(s)$, $se\in L(S/G)$. As $L(S/G)=\overline K$, $se\in \overline K$. Therefore, $\Gamma(s)=\Gamma_K(s)$ and condition~\eqref{eq:FC-2} is satisfied.  
For (ii), $E_f^\text{cd}(w)\cap \Gamma(s)\neq\emptyset$. As $E_f^\text{cd}(w)\subseteq E_f$, $\Gamma(s)\cap E_f\neq \emptyset$. Let $e\in S(w)\cap \Gamma(s)$. By Definition~\ref{def:S/G}, $se\in L(S/G)=\overline K$. Hence, $e\in \Gamma_K(s)$ and $S(w)\cap \Gamma(s)\subseteq \Gamma_K(s)$. Let $e\in \Gamma_K(s)$. As $\overline K=L(S/G)$, $se\in L(S/G)$. By Definition~\ref{def:S/G}, $e\in S(w)$. Thus, $\Gamma_K(s)\subseteq S(w)\cap \Gamma(s)$. We have $\Gamma_K(s)= S(w)\cap \Gamma(s)=E^\text{cd}_f(w)\cap \Gamma(s)$. That is, condition \eqref{eq:FC-1} in Definition~\ref{def:FC} is satisfied.

We can conclude that for all $w\in P(\overline K)$ and $s\in P^{-1}(w)\cap \overline K$, the set $E^\text{cd}_f(w)$ satisfying 
condition~\eqref{eq:FC-1} or \eqref{eq:FC-2} in Definition~\ref{def:FC} exists. Therefore, $\overline K$ is forcing consistent w.r.t $G$, $E_f$, and $P$.
\end{proof}

Theorem~\ref{thm:S} establishes that forcing consistency is a necessary and sufficient condition for the existence of an event forcing supervisor under partial observation. Once a (regular) language $K$ is verified to be forcing consistency, the sets $E^{cd}_f(w)$ can be determined using the approach presented at the end of Section~\ref{sec:forcibility}. Accordingly, the partial-observation supervisor $S$ can be derived following equations~\eqref{eq:S(w)-1} and \eqref{eq:S(w)-2}.

\begin{example}\label{eg:theorem}[Continuing Example~\ref{eg:FC}]
    Consider again the plant and specifications $K_1$ and $K_2$ shown in Fig.~\ref{fig:NotClose}. As discussed in Example~\ref{eg:FC}, both $K_1$ and $K_2$ are forcing consistent. Based on equations~\eqref{eq:S(w)-1} and \eqref{eq:S(w)-2}, the supervisors that enforce $K_1$ and $K_2$ are defined respectively as: 
    $$S_1(w)=\left\{\begin{array}{ll}
        \{f_{u1}\} & \text{ if } w=\eps,\\
        \emptyset & \text{ otherwise}.
    \end{array}\right.$$
        $$S_2(w)=\left\{\begin{array}{ll}
        \{f,f_{u2}\} & \text{ if } w=\eps,\\
        \emptyset & \text{ otherwise}.
    \end{array}\right.$$
    Indeed, according to Definition~\ref{def:S/G}, the resulting closed-loop systems are $L(S_1/G)=\overline{K_1}$ and $L(S_2/G)=\overline{K_2}$, as illustrated in Figs.~\ref{fig:NotClose}(b) and (c). \hfill $\diamond$
\end{example}

\subsection{Property of forcing consistency}
Having established the necessary and sufficient condition for a supervisor existence, we now examine the properties of forcing consistency. In traditional supervisory control, properties such as controllability and observability often exhibit closure under union or intersection, which enables the synthesis of a unique supremal or infimal supervisor. However, forcing consistency behaves differently due to the preemptive nature of the forcing mechanism.

\begin{proposition}\label{prop:FC-closed}
Let $G$ be the plant, $E_f$ a set of forcible events, and $P$ the projection. Given two nonempty sublanguages $K_1,K_2\subseteq L(G)$ that are both forcing consistent w.r.t $G$, $E_f$, and $P$, 
   \begin{enumerate}
       \item[1.] $K_1 \cup K_2$ need \emph{not} be forcing consistent;
       \item[2.] $K_1 \cap K_2$ need \emph{not} be forcing consistent.
   \end{enumerate}
\end{proposition}
\begin{proof}
1. Consider plant $G$, $K_1$, and $K_2$ in Fig.~\ref{fig:NotClose}. According to Example~\ref{eg:FC}, $K_1$ and $K_2$ are both forcing consistent. However, $K_1\cup K_2$ is \emph{not} forcing consistent: To only preempt event $e_2$ after $s_2=f_{u2}$, event $f$ has to be forced when $\eps$ is observed, i.e., $E^\text{cd}_f(\eps)=\{f\}$. Then, for $s_1=f_{u1}\in P^{-1}(\eps)\cap (\overline{K_1\cup K_2})$, $\Gamma(s_1)\cap E^\text{cd}_f(\eps)\neq \emptyset$. However, $\Gamma_{K_1\cup K_2}(s_1)=\{f,e_1\}\neq E^\text{cd}_f(\eps)\cap\Gamma(s_1)=\{f\}$, i.e., condition~\eqref{eq:FC-1} is not satisfied.

\begin{figure}
    \centering
    \includegraphics[width=0.15\textwidth]{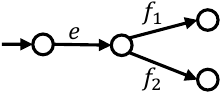}
    \caption{Plant $G$, where $E=E_o=\{e,f_1,f_2\}$ and $E_f=\{f_1,f_2\}$.}\label{fig:intersection}
\end{figure}

2. Consider plant $G$ in Fig.~\ref{fig:intersection}. Let $K_1=L(G)\setminus\{ef_2\}$ and $K_2=L(G)\setminus\{ef_1\}$. According to Definition~\ref{def:FC}, both $K_1$ and $K_2$ are forcing consistent w.r.t $G$, $E_f$, and $P$:
\begin{itemize}
    \item For $K_1$, $\forall w\in P(\overline{K_1})$, sets $E_f^{\text{cd}}(w)\subseteq E_f$ satisfying conditions~\eqref{eq:FC-1} or \eqref{eq:FC-2} exist: 
    $$E_f^{\text{cd}}(w)=\left\{\begin{array}{ll}
    \{f_1\}, & \text{ if } w=e;\\
    \emptyset, &\text{ otherwise}.
    \end{array}\right.$$
    \item For $K_2$, $\forall w\in P(\overline{K_2})$, sets $E_f^{\text{cd}}(w)\subseteq E_f$ satisfying conditions~\eqref{eq:FC-1} or \eqref{eq:FC-2} exist:
    $$E_f^{\text{cd}}(w)=\left\{\begin{array}{ll}
    \{f_2\}, & \text{ if } w=e;\\
    \emptyset, &\text{ otherwise}.
    \end{array}\right.$$
\end{itemize}
However, their intersection $K=K_1\cap K_2=L(G)\setminus\{ef_1,ef_2\}$ is not forcing consistent because, for $w=e$ no subset of $E_f$ satisfies either condition~\eqref{eq:FC-1} or \eqref{eq:FC-2}.  
\end{proof}

This proposition highlights a significant departure from classical results. Unlike observability, which is closed under intersection but not under union, forcing consistency is closed under neither union nor intersection. This result may seem counter-intuitive at first glance. It is important to note that the failure of closure under intersection holds even when all events are observable: in the example used to prove the second part of the theorem, all events are observable. This leads to another result that \emph{forcibility is not closed under intersection}. In other words, even with perfect information, the intersection of two forcing decisions may lead to a violation of the consistency preemption. Consequently, a unique supremal or infimal forcing consistent sublanguage of $L(G)$ generally does not exist. By Theorem~\ref{thm:S}, this implies that for a given specification $K$, there generally exist multiple, mutually incomparable maximally permissive forcing consistent supervisors.

\begin{example}\label{eg:notclose}[Continuing Examples 3.6 and 3.8]
    Consider again the plant in Fig.~\ref{fig:NotClose}(a) and $K=L(G)\setminus\{f_{u2}e_2\}$. As demonstrated in Example~\ref{eg:test}, $K$ is not forcing consistent. Nevertheless, $K_1$ and $K_2$ from Example~\ref{eg:FC} are two maximal forcing consistent sublanguages of $K$. Accordingly, the two supervisors presented in Example~\ref{eg:theorem} are two mutually incomparable maximally permissive supervisors for $K$. \hfill $\diamond$
\end{example} 

\section{Conclusions and Future Work}\label{sec:conclusion}
In this paper, we investigated the supervisory control of discrete-event systems with event forcing under partial observation. We demonstrated that the classical notion of forcibility is insufficient in this context because the supervisor's forcing actions can have unintended preemptive effects on indistinguishable strings. To address this, we proposed the notion of \emph{forcing consistency} and proved it to be the necessary and sufficient condition for achieving a given specification. Notably, we found that forcing consistency is not closed under union or under intersection, implying that a unique supremal or infimal forcing-consistent sublanguage does not generally exist.

Future work will focus on two directions. First, we aim to develop efficient algorithms for computing maximal forcing-consistent sublanguages, providing a practical tool for supervisor synthesis. Second, we will investigate a stronger property than forcing consistency that is closed under union, which would allow for the identification of a unique optimal controlled behavior. Finally, traditional controllability and observability will be included.

\bibliographystyle{unsrt}
\bibliography{EventForcing}

\end{document}